\documentclass{article}

\usepackage{arxiv}

\usepackage[utf8]{inputenc} 
\usepackage[T1]{fontenc}    
\usepackage{hyperref}       
\usepackage{url}            
\usepackage{booktabs}       
\usepackage{amsfonts}       
\usepackage{nicefrac}       
\usepackage{microtype}      
\usepackage{lipsum}

\usepackage{relsize,balance,lipsum,bbm,enumerate,times,comment,color,graphicx,setspace,mathdots,mathrsfs,amssymb,latexsym,amsfonts,amsmath,cite,stmaryrd,caption,pgf,tikz,accents,mathtools,tabu,enumitem,hhline,array,epstopdf,nicefrac,amsthm,microtype,algorithmic,array,float,bm,url}

\usepackage{graphicx}
\graphicspath{ {./images/} }

\usepackage{relsize,balance,lipsum,bbm,enumerate,times,comment,color,graphicx,setspace,bbm,mathdots,mathrsfs,amssymb,latexsym,amsfonts,amsmath,cite,stmaryrd,caption,pgf,tikz,accents,mathtools,tabu,enumitem,hhline,array,epstopdf}
\usepackage[ruled,vlined]{algorithm2e}
\usepackage{multirow}
\usepackage[utf8]{inputenc}
\graphicspath{{./images/}}
\usepackage[font=small,skip=0pt]{caption}
\usepackage{amssymb,multicol}
\usepackage{enumerate}
\usepackage{dirtytalk}
\usepackage{makecell}
\allowdisplaybreaks
\usepackage{amsthm}

\usepackage{amsthm}

\title{Quantification of Disaggregation Difficulty with Respect to the Number of Meters}

\author{
    Elnaz~Azizi\\
    Dept. of Electrical and Computer Engineering\\
    Tarbiat Modares University\\
    Tehran, Iran\\
    \texttt{e.azizi@modares.ac.ir}
\And
    Mohammad~T~H~Beheshti\\
    Dept. of Electrical and Computer Engineering\\
    Tarbiat Modares University\\
    Tehran, Iran\\
    \texttt{mbehesht@modares.ac.ir}
\And
    Sadegh~Bolouki\\
    Dept. of Electrical and Computer Engineering\\
    Tarbiat Modares University\\
    Tehran, Iran\\
    \texttt{bolouki@modares.ac.ir}
}

\theoremstyle{definition}

\newtheorem{theorem}{Theorem}

\begin{document}
\maketitle

\begin{abstract}

    A promising approach toward efficient energy management is non-intrusive load monitoring (NILM), that is to extract the consumption profiles of appliances within a residence by analyzing the aggregated consumption signal. Among efficient NILM methods are event-based algorithms in which events of the aggregated signal are detected and classified in accordance with the appliances causing them. The large number of appliances and the presence of appliances with close consumption values are known to limit the performance of event-based NILM methods. To tackle these challenges, one could enhance the feature space which in turn results in extra hardware costs, installation complexity, and concerns regarding the consumer's comfort and privacy. This has led to the emergence of an alternative approach, namely semi-intrusive load monitoring (SILM), where appliances are partitioned into blocks and the consumption of each block is monitored via separate power meters.
    
    While a greater number of meters can result in more accurate disaggregation, it increases the monetary cost of load monitoring, indicating a trade-off that represents an important gap in this field. In this paper, we take a comprehensive approach to close this gap by establishing a so-called notion of ``disaggregation difficulty metric (DDM),'' which quantifies how difficult it is to monitor the events of any given group of appliances based on both their power values and the consumer's usage behavior. Thus, DDM in essence quantifies how much is expected to be gained in terms of disaggregation accuracy of a generic event-based algorithm by installing meters on the blocks of any partition of the appliances. Experimental results based on the REDD dataset illustrate the practicality of the proposed approach in addressing the aforementioned trade-off.
\end{abstract}
\section{Introduction}
In-depth studies on the residential energy management strategies show that appliance load monitoring and detailed information about the consumption pattern of consumers play decisive roles in reducing overall consumption \cite{gopinath2020energy}. The most accurate way to extract this information is intrusive load monitoring (ILM), in which a sensory device is installed on each appliance of a house to measure, record and report its consumption \cite{yang2020event}. However, installing a meter on each individual appliance of residential buildings is expensive and time consuming and invades the privacy of the household \cite{liu2019low}. 

To tackle these challenges, non-intrusive load monitoring (NILM) was first proposed in 1992 by Hart \cite{hart1992nonintrusive}. NILM aims to extract the consumption profile of each appliance from the aggregated signal measured by existing meters via purely analytical methods. Since NILM does not interfere with the existing building infrastructure, it is more practical than ILM methods and cost and time-efficient \cite{zhang2019new}. Various NILM methods have been proposed based on different features of appliances \cite{henao2015approach}. In particular, considering the existing smart meters, the active power data is more accessible than other features such as reactive power data. Thus, the interest in the NILM research based on the active power of appliances as their specific features
have grown over the past decade \cite{mueller2016accurate}.

Active power-based NILM studies can be roughly divided into two main classes, event-based methods and sample-based methods \cite{mueller2016accurate}. Event-based methods are focused on detecting significant variations in the power signal called events and classifying them as per the specific mode transitions of appliances causing them \cite{lu2019hybrid}. The general idea behind sample-based NILM methods is that the consumption behavior of appliances can be represented as a finite-state machine and the total power signal can be disaggregated based on the learned-model of state transitions of appliances \cite{zhao2018improving}. Compared to the sample-based methods, event-based methods are more direct and comprehensive and often have lower computational complexities. Therefore, they have attracted an increasing attention in recent years \cite{rehman2019event}.

While many event-based NILM approaches have been developed and validated on small groups of appliances, the performance of these methods decreases with increasing the number of appliances \cite{tang2015distributed} and in the presence of multi-mode appliances \cite{ egarter2014paldi} or appliances with close power values \cite{liu2016non}. To address these challenges, conventional NILM approaches have been focused on feature space improvements such as utilizing light \cite{srinivasan2013fixturefinder}, ON/OFF duration of appliances \cite{dinesh2019residential}, and weather and temperature features \cite{ma2017toward}. However, the use of additional meters to record environmental features or consumer's usage behavior requires manual effort and time in addition to the high costs \cite{gopinath2020energy}. Moreover, the accuracy of these approaches highly depends upon the training dataset, and achieving high accuracy requires a high volume of training dataset, gathering which raises privacy concerns for consumers \cite{machlev2018modified}. 


These challenges lead to a compromised approach toward appliance load monitoring, the so-called semi-intrusive load monitoring (SILM), which extracts the consumption of appliances from multiple aggregated signals \cite{xu2018classifier}. In other words, in the SILM approach, the set of appliances of a building is divided into multiple subgroups and a meter is installed to measure the aggregated signal of appliances in each subgroup \cite{dash2019semi}. A load disaggregation problem should then be solved for each subgroup. SILM methods inherit the advantages of ILM and NILM and potentially allow for an  optimal trade-off between the computational complexity, accuracy and cost \cite{hosseini2017non, tang2015distributed}. Evaluating the performance of SILM and NILM with real datasets prove that utilizing a very limited number of meters, the SILM approach significantly improves the accuracy of power disaggregation for large-scale and multi-mode appliances \cite{langevin2020crosstalk}.

Taking into account the limitations presented by the electrical network topology, main challenges of SILM are loosely described as obtaining a minimal number of meters and proper partitioning of appliances which make possible extracting the consumption profiles of individual appliances with high accuracy \cite{tang2015distributed}. Besides, the types and number of appliances, the infrastructure of residential buildings that imposes constraints and shrinks the set of feasible subgroups of appliances \cite{dash2019novel}, and the willingness of consumers to pay for the meters are not the same across the board and must be taken into consideration \cite{xu2018classifier}. Therefore, the proper number of meters should be determined individually for each consumer. 
An optimal number of meters, as well as the corresponding partition of the appliances, in essence depend on how distinguishable the appliances are with respect to their power consumption.
In \cite{pochacker2015proficiency}, authors quantified the distinguishability of the set of appliances assuming fixed power values associated with operation modes of appliances. However, due to voltage fluctuations in the grid, the consumption of an appliance is not fixed.

In this paper, we formulate the so-called {\it disaggregation difficulty metric} (DDM) in part based on the distribution of the active power of appliances. A second factor that also plays an important role in the DDM formulation is the consumer's usage behavior. For instance, appliances that are rarely used would contribute less to the DDM. In the second phase of this study, to pursue a scalable solution for obtaining an optimal number of meters for each consumer, we quantify the relationship between the number of meters, or meters' cost, and the DDM. 
It is worthwhile to note that the two factors serving as the bases of the DDM formulation are not immediately available and have to be extracted from the training dataset consisting of discrete time-series data of the consumption signal of type I appliances (those with only two operation modes (ON/OFF modes)) or type II appliances (those with more than two operation modes). Fig. \ref{framework} demonstrates the entire framework of this paper and the significance of our contributions is highlighted below.
\begin{figure}[t]
    \centering
    \includegraphics[width=.5\linewidth]{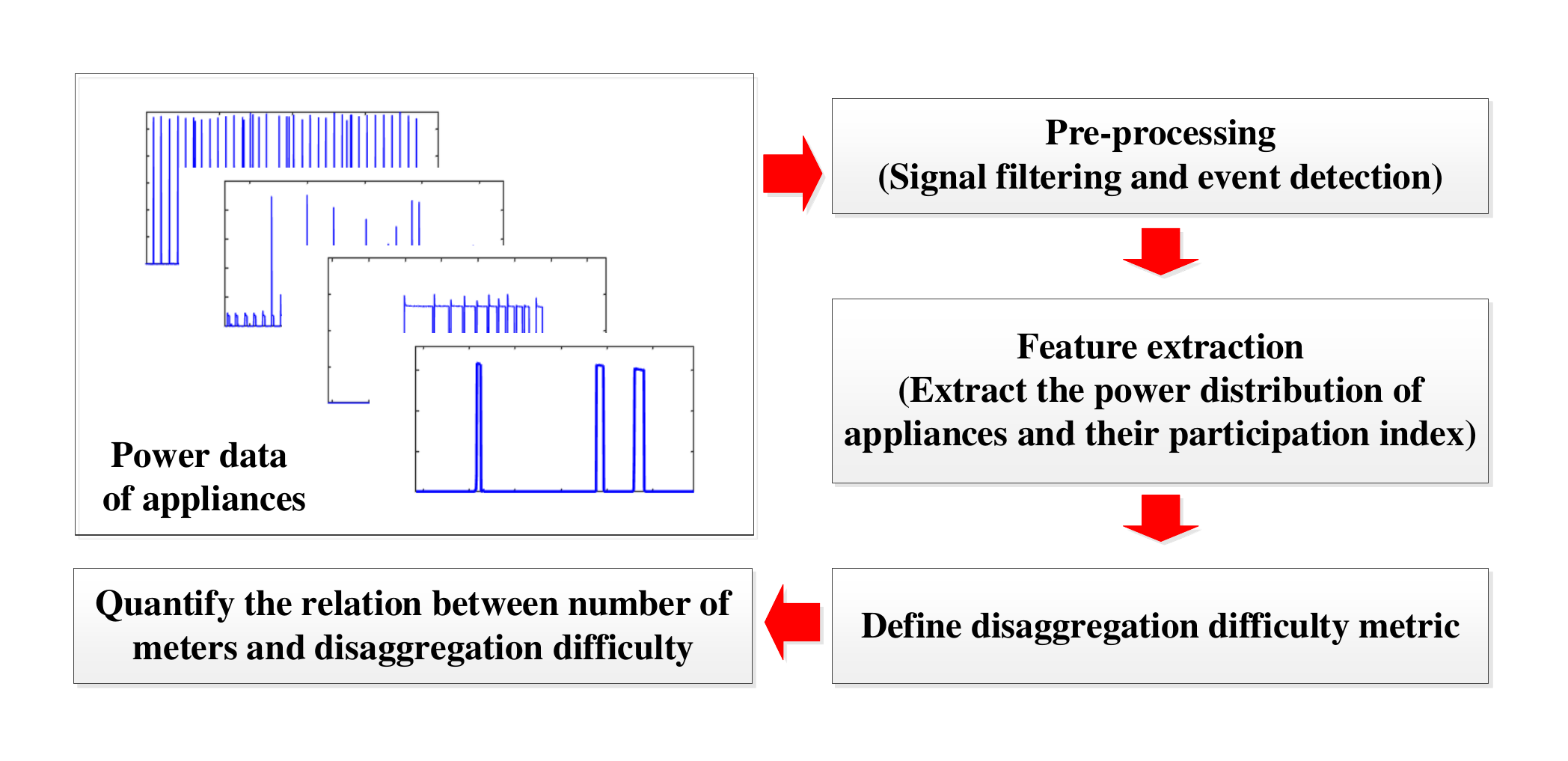}
    \caption{Framework of this paper}
    \label{framework}
    \vspace{-.1in}
\end{figure}
\begin{itemize}
    \item To the best of our knowledge, this is the first paper to quantify the difficulty of distinguishing of appliances for event-based NILM classification methods.
    \item Instead of considering a fixed value for the power consumption of an appliance in each of its operation modes, we extract a distribution of power from the training dataset, which makes the results more realistic.
    \item In contrast to \cite{pochacker2015proficiency}, the DDM formulation also depends on the usage behavior of consumers rather than being rigidly dependent upon the active power of appliances.
    \item The trade-off between disaggregation difficulty and the number of meters is illustrated in this paper. 
    Consequently, under the cost constraints and infrastructural limitations of the building of interest, an optimal number of meters and the associated partition of appliances can be obtained for the most accurate load monitoring possible.
    \item Low-frequency data is used for feature extraction, which meets practical considerations and also leads to a load monitoring process with a relatively low computational complexity.
    \end{itemize}




The remainder of this paper is organized as follows. Section~\ref{Pre-Processing} describes the pre-processing phase and feature extraction.
Section~\ref{Dis def} then conceptualizes and formulates the DDM. Afterwards, in Section~\ref{Relation}, the relationship between the number of meters and the DDM is made explicit. Finally, the effectiveness of the proposed metric is evaluated in Section~\ref{Simulatin Study} using the Reference Energy Disaggregation Data Set (REDD) \cite{kolter2011redd}, before Section~\ref{Conclusion} concludes the paper.

\section{Pre-processing and Feature Extraction}\label{Pre-Processing}
 %
 In this section, filtering the power signals and feature extraction will be discussed first.
 Then, a clustering-based method is utilized to obtain probable mode transitions of each appliance, which will be described in Subsection~\ref{mode extraction}. Based on these results, two main features required for the DDM formulation are extracted in Subsection~\ref{Feature extraction}.
 


\subsection{Signal Filtering and Event Detection} \label{signal filtering}
Generally speaking, event detection can be defined as the process of identifying the possible occurrence of appliances' operating mode transitions in the aggregated signal or signal of each appliance \cite{alcala2017event}. In other words, an event in a signal is defined as the transition from a steady state to a different steady state \cite{xu2018classifier}. Event detection methods relying on fixed thresholds do not perform well for different datasets and should be adjusted for different sets of appliances based on their power values \cite{ma2017toward}. In this paper, we utilize our previously proposed method of event detection in \cite{azizi2020novel}, which is free of fixed thresholds and characterizes an event as an \say{uncommon} value change in the signal. In the first stage of this algorithm, the min/max ratio between any two consecutive sampled values of the power signal ($P(t)$) is calculated. Then, they are subtracted from 1, saved in a vector $M$ and the standard deviation of $M$ is computed. Finally, for any $t$, if $M(t)$ is greater than the calculated standard deviation, it is considered as an outlier and $t$ is saved as the outlier occurrence instance in vector $M_o$ of outlier instances.

\subsection{Extracting Appliances' Mode Transitions}\label{mode extraction}
As opposed a majority of the existing literature that extracts an appliance's operation modes and their respective power values visually using samples of the appliance's power signal in the training set, or based off of the spec sheet of the appliance, we utilize a clustering-based approach for extracting plausible mode transitions of appliances and respective power values from the training dataset. Generally speaking, clustering is an unsupervised learning technique that aims to properly partition data samples of a dataset into a number of groups called clusters based on their similarity. A commonly used method of clustering is the $K$-means algorithm which is fast in comparison with other methods and has decent performance in segmenting high volume of data points \cite{alpaydin2020introduction}. Its procedure is based on minimizing $f_1$ defined as 
        \begin{equation}\label{costfun}
            f_1 = \sum_{j=1}^K \sum_{i: z_i \in C_j} \left\| z_i - \bar{z}^j \right\|_2^2, \bar{z}^j = \frac{1}{|C_j|}\sum_{i:z_i \in C_j} z_i,
        \end{equation}
    where $z_i$'s denote the data points, $C_j$'s are the clusters, $\bar{z}^j$'s represent their centroids, and $K$ is the number of clusters. A fundamental part of this algorithm is selecting the number $K$ of clusters. A large number $K$ renders the clustering inconsequential and results in a high computational complexity of the ensuing process that would otherwise benefit from a meaningful clustering, while a small number $K$ leads to clusters with dissimilar members. In this paper, to select a proper number $K$, the elbow method is applied to the events dataset \cite{azizi2020residential}. This method computes $f_1$ by gradually increasing the number of clusters and when the segmenting cost varies a little in comparison with the next step, it stops.

\subsection{Feature Extraction} \label{Feature extraction}
Two main features of appliances that affect the performance of a generic event-based disaggregation method are (i) their power consumption at their operation modes, and (ii) their participation in the aggregated power consumption, which is dependent upon the consumer's usage behavior. Previous research on computing the difficulty of disaggregation, \cite{pochacker2015proficiency}, considered fixed values for power consumption of appliance modes, which is unrealistic due to the uncertainties in the grid. Moreover, \cite{pochacker2015proficiency} does not take the consumer's usage behavior of appliances into account. In the following, we extract the aforementioned two features from the training dataset.
\subsubsection{Transitions' Power Distributions}
The consumption power of appliances in each operation mode is subject to various uncertainties of significant magnitudes such as frequent and large voltage fluctuations caused by the variability of renewable energies, thermal noise, and many other physical parameters \cite{henao2016active}. All these uncertainties cause power fluctuations with finite variance. We assume that the power value associated with appliance $i$'s mode transition $j$ can be approximated by a Gaussian distribution \cite{7605545},
\begin{equation}\label{dist}
    f_{ij}(x,\mu_{ij},\sigma_{ij}^2)=\frac{1}{\sqrt{2\pi}\sigma_{ij}^2}e^{-\frac{(x-\mu_{ij})^2}{2\sigma_{ij}^2}},
\end{equation}
where $\mu_{ij}$ and $\sigma_{ij}^2$ are the mean and variance, respectively, calculated given the data points available for that transition. This assumption is made for a neater presentation of the results and, as will be discussed in Subsection~\ref{discussion}, can be removed.

\subsubsection{Computing Participation Indices} \label{P index}
It should be noted that appliances with overlapping power values are not used with the same frequency. Thus, in addition to the appliances' power consumption in their operation modes, their participation in the aggregated power signal also closely affects the disaggregatability of the aggregated signal and the accuracy of an event-based load disaggregation method. In other words, appliances that are rarely in use should have less effect on the disaggregation difficulty of the aggregated signal. 

In this regard, a novel index called \textit{participation index} is defined for each mode transition of the appliances. Given a mode transition $j$ of an appliance $i$, this index, denoted by $\pi_{ij}$, is calculated as
    \begin{equation}\label{p_co}
        \pi_{ij} = \frac{Num_{T_{ij}}}{Num_{TOT}},
    \end{equation}
where $Num_{T_{ij}}$ and $Num_{TOT}$ stand for the number of events caused by this specific mode transition and the total number of events of aggregated signal, respectively. Since each event of the aggregated signal is caused by a unique mode transition, it should be clear that
    \begin{equation}\label{prob1}
        \sum_{i=1}^{N}\sum_{j=1}^{t_i}\pi_{ij}=1,
    \end{equation}
where $t_i$ denotes the number of probable mode transitions of appliance $i$. Finally, one intuitively expects that appliances with close power values and close participation indices significantly complicate the disaggregation process.
\section{Disaggregation Difficulty Metric Formulation}\label{Dis def}

    Formulating the DDM is carried out in three steps, (i) obtaining the probability distributions of events in the aggregated signal given the distribution of the power value associated with a generic event of each appliance and its participation index, (ii) calculating the probability that a sample event is caused by a given appliance, and (iii) characterizing the DDM based on these probability distributions and probabilities.

\subsection{Distribution of a Generic Event in Aggregated Signal}
In the process of formulating the DDM, we need to generate a sample event of the aggregated signal. This is carried out considering the conditional distribution of the aggregated signal given the power distribution of each appliance's mode transition and its participation index as obtained in Subsection~\ref{P index}. By applying the Bayes’ rule, if $T$ represents a generic event, the probability distribution function of its value is obtained as
\begin{equation}\label{prob}
    f_T(\alpha)=\sum_{i=1}^{N}\sum_{j=1}^{t_i}\pi_{ij}f_{ij}(\alpha),
\end{equation}
where $N$ and $t_i$ stand for the total number of appliances and total number of transitions of appliance $i$, respectively; $\pi_{ij}$ shows the participation index of $j$th transition of appliance $i$; and $f_{ij}$ is the Gaussian distribution of the associated power value of transition $j$ of appliance $i$. Recalling \eqref{prob1} and noticing that $\int f_{ij}(\alpha)d\alpha=1$ for any $i,j$, one verifies that $\int f_T(\alpha)d\alpha=1$.

As an example, we consider a simple experiment with 3 appliances and assume that the characteristics of their Gaussian power distribution and their participation indices are extracted based on \eqref{dist} and \eqref{p_co} and given in Table~\ref{Modes}. Fig.~\ref{test} then demonstrates the probability distribution function by which the value of a generic event of the aggregated signal is generated.
\begin{table}[t]\centering
\caption{Power distributions of transitions}
\begin{tabular}{|c|c|c|c|c|c|c|}
\Xhline{2\arrayrulewidth}
Apps. & \multicolumn{3}{c|}{Transition 1} & \multicolumn{3}{c|}{Transition 2} \\ \Xhline{2\arrayrulewidth}
           & $\mu$      & $\sigma$   & $\pi$   & $\mu$      & $\sigma$   & $\pi$    \\ \hline
$1$          & 210     & 10   & 0.193    & 400    & 20   & 0.096 \\ \hline
$2$          & 220     & 12  & 0.322    & 1080     & 30  & 0.258 \\ \hline
$3$          & 1100    & 40   & 0.129  & -       & -  & - \\ \Xhline{2\arrayrulewidth}
\end{tabular} 
\label{Modes}
\end{table}
\begin{figure}[t]
    \centering
    \includegraphics[width=.5\linewidth]{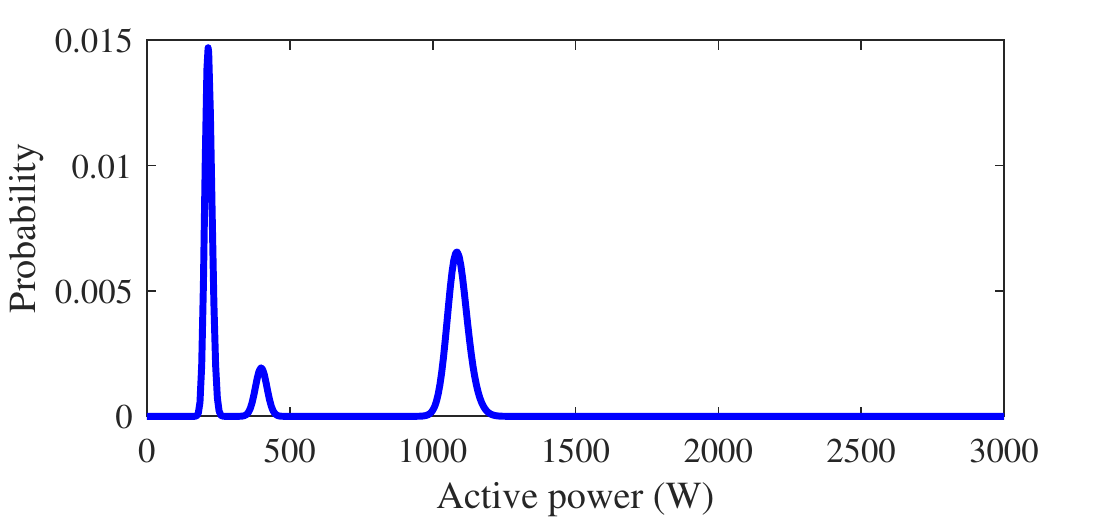}\vspace{.1in}
    \caption{Probability distribution of the value of a generic event for the example of Table~\ref{Modes}.}
    \label{test}
\end{figure}

\subsection{Probability-based Disaggregation}\label{Sec_Probability}

    Given an event $T$ with value $\alpha$, we now calculate the probability that $T$ has been caused by mode transition $j$ of appliance $i$ as
\begin{equation}\label{belon3}
 P(T_{ij}\rightarrow T\,|\,T=\alpha)=\frac{\pi_{ij}f_{ij}(\alpha)}{f_T(\alpha)}.
\end{equation}
For the example of Table~\ref{Modes}, Fig.~\ref{belong} depicts for different values of $\alpha$, the probability that it is caused by each appliance's mode transition. Is should be clear in Fig.~\ref{belong} that, power values that correspond to unique mode transitions in Fig.~2, such as 400W, bear little uncertainty as to which mode transition causes them. However, the overlapping distributions of $T_{22}$ and $T_{31}$ means that an example event with value of 1090W is difficult to disaggregate as it may be caused by either of the two mode transitions with fairly high probability.
\begin{figure}[t]
    \centering
    \includegraphics[width=.5\linewidth]{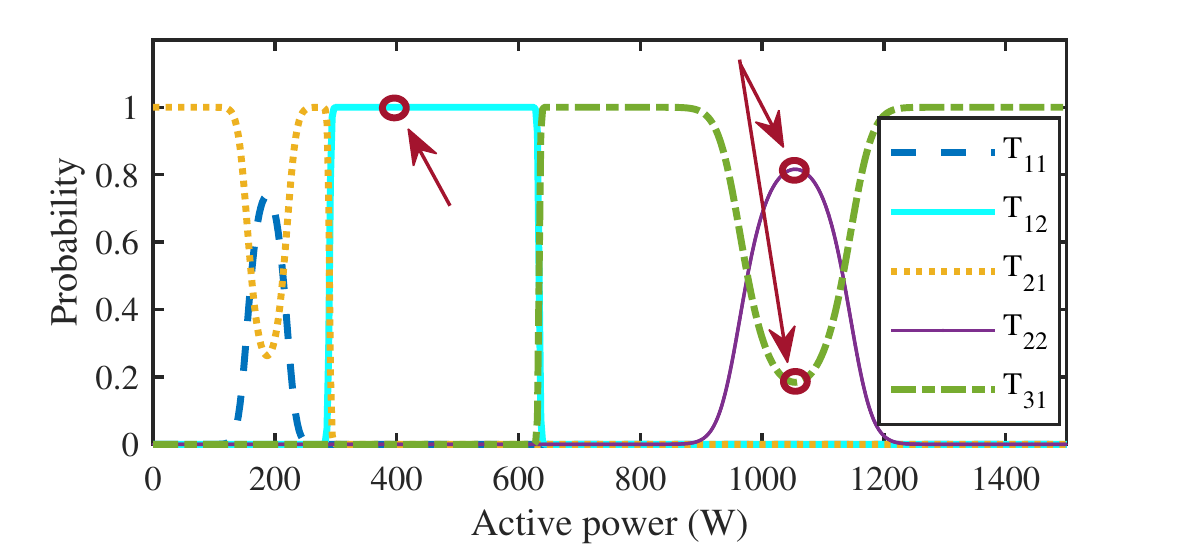}\vspace{.1in}
    \caption{Probability that a given event is caused by different mode transitions, with respect to the value of the event. Arrows point to the event values 400W and 1090W.}
    \label{belong}
\end{figure}
\subsection{Disaggregation Difficulty Metric (DDM)}
    Among various types of methods, the concept of entropy has been used as a viable measure for the average uncertainty of random variables with discrete values. In the following, we aim to intuitively introduce the concept of entropy in our case. We recall that $P(T_{ij}\rightarrow T\,|\,T=\alpha)$ is the conditional probability that mode transition $j$ of appliance $i$ has caused a given event with the value $\alpha$. For any $\alpha$, these conditional probabilities form a probability distribution, the entropy of which quantifies the uncertainty as to which mode transition causes an event with the value $\alpha$, that is
\begin{equation}\label{DD}
 e_\alpha=-\sum_{i=1}^N\sum_{j=1}^{t_i} p(T_{ij}\rightarrow T\,|\,T=\alpha)\log (T_{ij}\rightarrow T\,|\,T=\alpha)),
\end{equation}
where $t_i$ is the total number of transitions of appliance $i$. The DDM can now be defined as the average uncertainty over all values of $\alpha$, that according to \eqref{prob} and \eqref{DD} is
\begin{equation}\label{8}
    DDM = \int f_T(\alpha) e_\alpha d\alpha.
\end{equation}
One notices that \eqref{8} describes a weighted average of uncertainty where each $e_\alpha$ is understandably weighted by the probability that the event takes the value $\alpha$. The DDM captures the challenge of disaggregating the aggregated power signal of a set of appliances. This metric is expected to confirm that the closer power values of different appliances in the set, the more difficult it is to extract their power profiles from the aggregated signal. For the example of Table~\ref{Modes}, the DDM equals 0.24. However, revising the parameters of the probability distributions associated with the transitions of appliance 2 to $(\mu_{21},\pi_{21})=(600,6)$ and $(\mu_{22},\pi_{22})=(2000,8)$, no pair of transitions with significantly overlapping power distributions remains and one obtains that $DDM=0.09$.

\section{Relationship Between the Number of Meters and the DDM} \label{Relation}

    The most accurate but costly process of load monitoring involves ILM, that is carried out by installing a meter on each appliance. At the other end of the spectrum stands NILM, which leads to the least accurate yet cheapest load monitoring. Within the spectrum between these two monitoring regimes are those based on the SILM approach, which aims at achieving an optimal trade-off between the overall cost of monitoring and the monitoring accuracy by partitioing various appliances into groups, herein called {\it blocks}, so as to monitor every block with a single, exclusive meter. This optimal trade-off problem can be translated to obtaining an optimal number of meters since it has a direct relationship with both the overall cost and the monitoring accuracy. While it is somewhat straightforward to estimate how much an extra meter adds to the overall cost, its contribution to the monitoring accuracy is not at all clear and calls for an extension of the DDM to the SILM setting.
    


As mentioned above, any SILM framework involves partitioning the appliances into blocks. Fig.~\ref{Partitioning} demonstrates all 5 possible partitions of 3 appliances for different numbers of blocks (meters). For information regarding the number of partitions of a set of size $N$, widely known as the Bell number $B_N$, the interested reader is referred to \cite{aigner1999characterization}.
\begin{figure}[t]
    \centering
    \includegraphics[width=.5\linewidth]{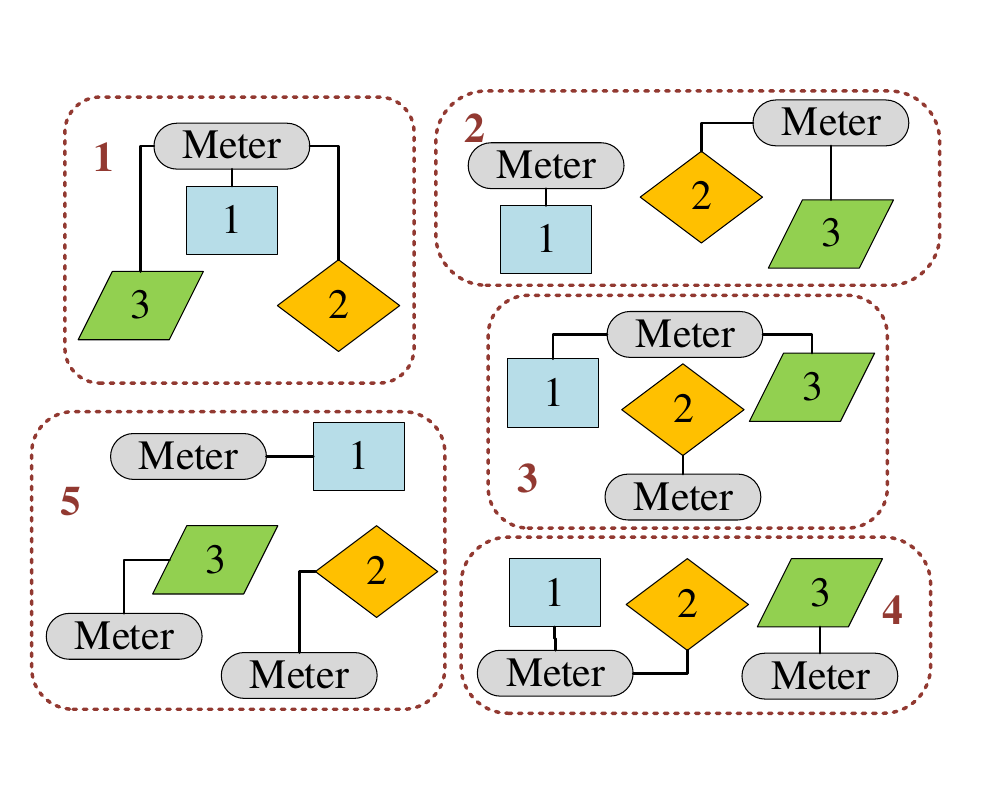}\vspace{-.1in}
    \caption{All partitions of 3 appliances with one, two and three meters}
    \vspace{-.15in}
    \label{Partitioning}
\end{figure}
In the rest of this section, we extend the DDM notion to the SILM setting, argue its well-definedness, and discuss how it can be utilized to obtain an optimal number of meters. Let $S_1,\ldots,S_b$ form a partition of appliances in the SILM paradigm, meaning that a meter has been assigned to measure the aggregated power of the appliances within each block $S_k$, $1 \leq k \leq b$. Given this partition, the DDM is calculated as follows. First, we focus on a single but arbitrary block, say $S_k$, and assume that a sample event of the value $\alpha$ has been caused by an appliance in this block. In view of \eqref{DD}, the following value captures the uncertainty as to which appliance of $S_k$ has caused this event:
\begin{align}
    e_{\alpha}(S_k) =
    & - \sum_{i\in S_k}\sum_{j=1}^{t_i}\Bigg[\left(\frac{p(T_{ij}\rightarrow T\,|\,T=\alpha)}{p(S_k\rightarrow T\,|\,T=\alpha)}\right)\nonumber\\
    &\hspace{.78in}\log\left(\frac{p(T_{ij}\rightarrow T\,|\,T=\alpha)}{p(S_k\rightarrow T\,|\,T=\alpha)}\right)\Bigg],\label{e_alpha(S_k)}
\end{align}
where
\begin{equation}\label{sk}
    p(S_k\rightarrow T\,|\,T=\alpha) = \sum_{i\in S_k}\sum_{j=1}^{t_i}p(T_{ij}\rightarrow T\,|\,T=\alpha)
\end{equation}
is the probability that a given event of the value $\alpha$ in the aggregated signal has been caused by the block $S_k$. For one to make sense of \eqref{e_alpha(S_k)}, we point out that the value
\begin{equation}
    \frac{p(T_{ij}\rightarrow T\,|\,T=\alpha)}{p(S_k\rightarrow T\,|\,T=\alpha)}= p(T_{ij}\rightarrow T\,|\,T=\alpha,\,S_k\rightarrow T),
\end{equation}
where $i \in S_k$, is the probability that a given event of the value $\alpha$ caused by an appliance in $S_k$ has been in fact caused by transition $j$ of appliance $i$.

Now, a weighted average over $e_\alpha(S_k)$ for $1 \leq k \leq b$ well captures the uncertainty associated with a given event of value $\alpha$ in the aggregated signal as to which appliance has caused it, that is
\begin{equation}\label{new e_alpha}
    e_\alpha \triangleq \sum_{k=1}^{b}  p(S_k \rightarrow T\,|\,T=\alpha) e_\alpha(S_k).
\end{equation}

Finally, having obtained the value $e_\alpha$ for any $\alpha$, the DDM is once again defined as
\begin{equation}\label{8 again}
    DDM = \int f_{T}(\alpha)e_\alpha d\alpha,
\end{equation}
where $f_T(\alpha)$ is defined in \eqref{prob} and remains the same for the SILM case, while $e_\alpha$ is defined in \eqref{new e_alpha}.

    We now argue that this extension of the DDM to the SILM paradigm is in line with an important, intuitive expectation one would have of such a notion. More specifically, one expects that if a partition is finer than another partition, that is if each of its blocks is contained in some block of the other partition, its corresponding DDM should be less than that of the other partition. In other words, further splitting blocks of a partition and adding extra meters accordingly should reduce disaggregation difficulty. In the following theorem, we state that our SILM-extension of the DDM in fact meets this expectation.

\begin{theorem}
    If a partition of the appliances is finer than another partition, its corresponding DDM is less than or equal to that of the other partition.
\end{theorem}

\begin{proof}

    We only prove the theorem for the case of an arbitrary two-block partition, which is finer than the single-block partition (the NILM case). The general case can be proved in a very similar fashion using induction, noticing that any set of finer blocks can always be obtained by sequentially splitting the original blocks. We start with the single-block partition and assume that $\{T_1,\dots,T_M\}$ represents the set of all transitions. In view of \eqref{DD}, for any value $\alpha$ we have
\begin{equation}\label{e1}
    e_\alpha =
    - \sum_{i=1}^{M} p(T_{i}\rightarrow T\,|\,T=\alpha)\log (p(T_{i}\rightarrow T\,|\,T=\alpha))
\end{equation}
Now, turning to the arbitrary two-block partition, we assume that the two blocks are in such a way that the transitions in the set $\{T_1,\dots,T_l\}$ are monitored with the first meter and the rest of transitions, i.e., $\{T_{l+1},\dots,T_M\}$, are monitored with the second meter. Now, in view of \eqref{new e_alpha}, while considering \eqref{e_alpha(S_k)} and \eqref{sk}, the new value of $e_\alpha$, denoted by $e'_\alpha$, is obtained as
\begin{align}
    e'_\alpha
    =
    &- \sum_{i=1}^{l} p(T_{i}\rightarrow T\,|\,T=\alpha)\log \frac{p(T_{i}\rightarrow T\,|\,T_\alpha)}{\sum_{j=1}^lp(T_{j}\rightarrow T\,|\,T_\alpha)}\nonumber\\
    &- \sum_{i=l+1}^{M} p(T_{i}\rightarrow T\,|\,T=\alpha)\log \frac{p(T_{i}\rightarrow T\,|\,T_\alpha)}{\sum_{j=l+1}^Mp(T_{j}\rightarrow T\,|\,T_\alpha)}
    \label{e2}
\end{align}
Since both denominators that appear in the right-hand-side of \eqref{e2} are less than or equal to 1, it is straightforward to conclude that $e'_\alpha \leq e_\alpha$. Noticing from \eqref{8} that the DDM corresponding to the single-block partition is an average over $e_\alpha$ with the weights $f_T(\alpha$), and according to \eqref{8 again}, the DDM for the two-block partition is an average over $e'_\alpha$ with the same weights, the proof is now complete.
\end{proof}
For the example of \ref{Sec_Probability}, DDM for different partitions is given in Table~\ref{T2}, which confirms that the most challenging disaggregation happens when in the NILM case where there is only one meter, while increasing the number of meters reduces the difficulty of disaggregation. In particular, as shown in Table~\ref{Modes}, appliance~2 has two operation modes, the first of which has overlapping power values with appliance~1, while the second of which has overlapping power values with appliance~3. This is also highlighted in Table~\ref{T2}, which implies that separating appliance~2 and directly monitoring its consumption significantly decreases the disaggregation difficulty to the point where it virtually achieves the ILM accuracy where three meters are utilized. Therefore, in summary, Table~\ref{T2} shows which 2-meter assignment is best with respect to the accuracy of disaggregation, while it deems the 3-meter assignment unnecessary as it increases the monitoring costs without meaningfully improving the monitoring accuracy. Table~\ref{T3} shows the minimum DDM for each number of meters and its corresponding cost based on \cite{xu2018classifier}. Consumers can choose the proper number of meters based on their willingness to pay for the meters' cost, the constraints imposed by the building infrastructure, and the DDM.

\begin{table}[t]\centering 
\caption{DDM for different partitions of the appliances}
\begin{tabular}{|c|l|c|}
\Xhline{2\arrayrulewidth}
Number of meters  & Meter assignment                                                             & $DDM$    \\ \Xhline{2\arrayrulewidth}
1                  & Meter1: 1,2,3                                                      & 0.24 \\ \hline
\multirow{3}{*}{2} & \begin{tabular}[l]{@{}l@{}}Meter1: 1\\ Meter2: 2,3\end{tabular}       & 0.17 \\ \cline{2-3} 
                   & \begin{tabular}[l]{@{}l@{}}Meter1: 2\\ Meter2:  1,3\end{tabular}       & 0.12 \\ \cline{2-3} 
                   & \begin{tabular}[l]{@{}l@{}}Meter1: 3\\ Meter2: 2,3\end{tabular}       & 0.18 \\ \hline
3                  & \begin{tabular}[l]{@{}l@{}}Meter1: 1\\ Meter2: 2\\ Meter3: 3\end{tabular} & 0.12 \\ \Xhline{2\arrayrulewidth}
\end{tabular}\label{T2}
\end{table}
\begin{table}[t]\centering
\caption{Cost and lowest achievable $DDM$ for different numbers of meters}
\begin{tabular}{|c|c|c|c|}
\Xhline{2\arrayrulewidth}
Number of meters & 1    & 2    & 3    \\ \Xhline{2\arrayrulewidth}
$DDM$              & 0.24 & 0.12 & 0.12 \\ \hline
Meter cost (\$)        & 200  & 400  & 600  \\ \Xhline{2\arrayrulewidth}
\end{tabular}\label{T3}
\end{table}
\section{Simulation Study} \label{Simulatin Study}
In this section, we make concrete the disaggregation difficulty concept using the REDD dataset \cite{kolter2011redd}.
This dataset consists of high, low-resolution power data. Features based on the high-frequency measurements are difficult to obtain due to the limitation of the storage space and the communication bandwidth \cite{xu2018classifier}.  
Seven appliances are considered in this paper, consisting of 5 type~I appliances (oven (OV), microwave (MW), kitchen outlets (KO), bathroom GFI (BGFI), and washer/drier (W/D)) and 2 type~II appliances (refrigerator (RFG) and dishwasher (DW)). The power consumption of these appliances in a day is illustrated in Fig.~\ref{3D_test_days}. 
\begin{figure}[t]
    \centering
    \includegraphics[width=.6\linewidth]{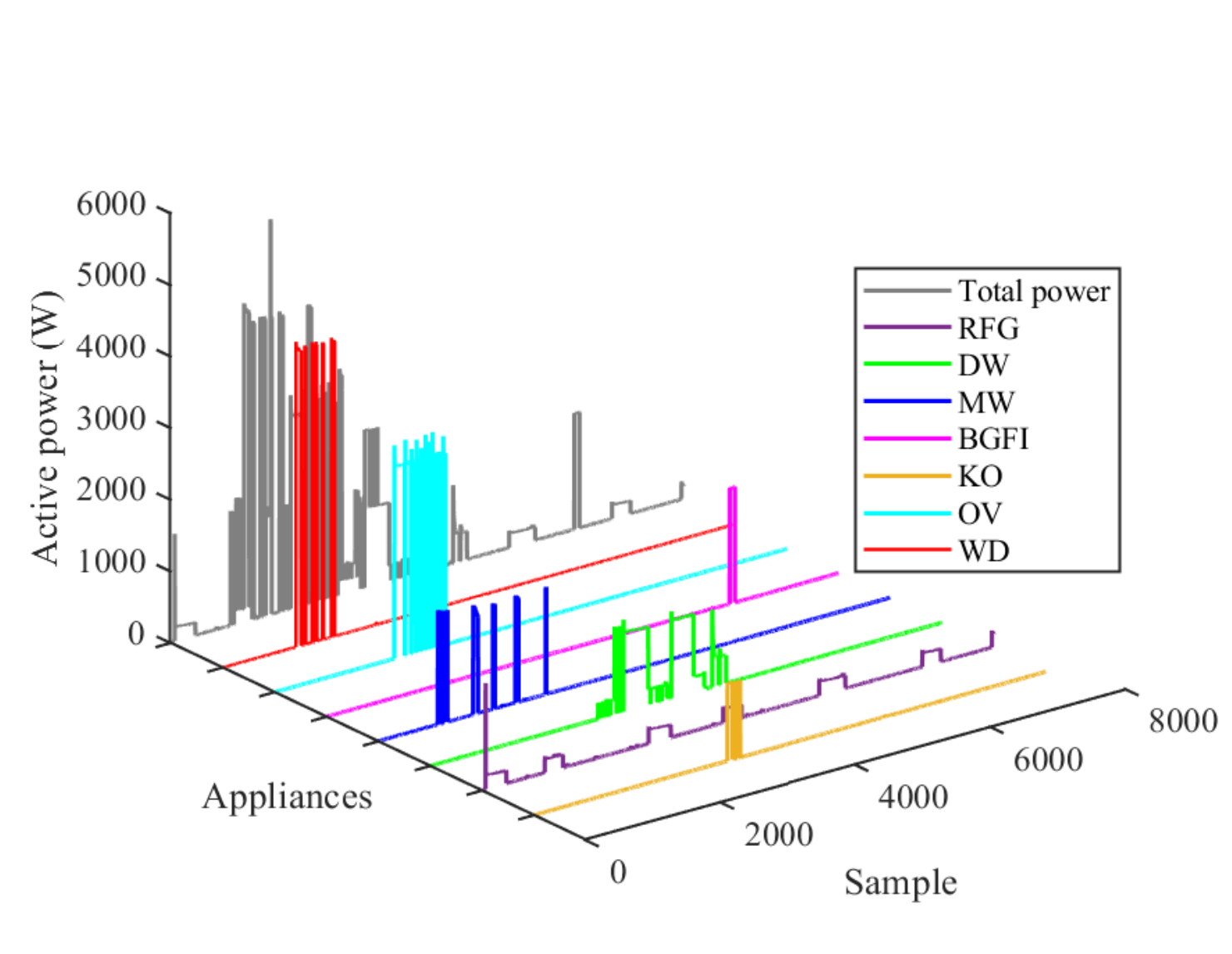}
    \caption{Power consumption of appliances vs the total power signal}
    \label{3D_test_days}
\end{figure}
In the following subsections, after filtering the signal of each appliance, detecting events, and obtaining the probable transitions of appliances, the power distribution of their transitions and their participation indices are extracted. Then, the DDM is computed based on the two extracted sets of features. Finally, the relationship between the DDM and the number of meters is made explicit.

\subsection{Pre-processing and Feature Extraction}

    To filter the power consumption signals of individual appliances in the training dataset and extract their respective events, the method described in Subsection~\ref{signal filtering} is utilized. 
    Then, to extract transitions of each appliance, the $k$-means algorithm is applied to the detected events of the signal of each appliance, and the proper number of transitions is obtained from the elbow method as described in Subsection~\ref{mode extraction}.
After transitions extraction, the power distribution of each transition is extracted based on \eqref{dist}. Then, based on \eqref{p_co}, the participation index of each mode of all appliances is calculated. Fig.~\ref{3Ddist} and Table~\ref{Modes} show the characteristics of the normal distribution of each transition of appliances and their participation indices.
\begin{figure}[t]
    \centering
    \includegraphics[width=.5\linewidth]{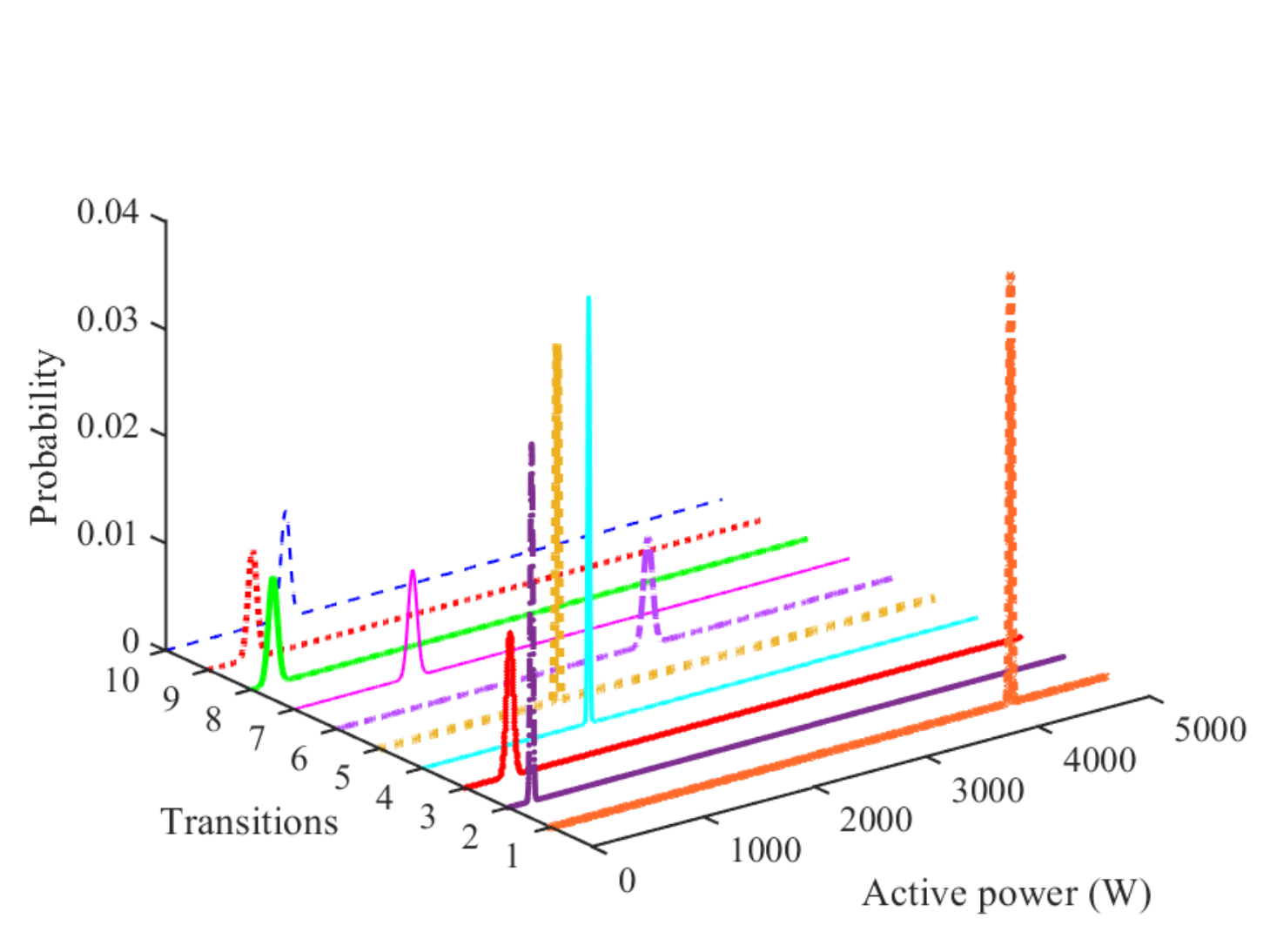}
    \caption{Power distribution of all transitions}
    \label{3Ddist}
\end{figure}
\begin{table}[t]\centering
\caption{Extracted features of appliances}
\begin{tabular}{|c|c|c|c|}
\Xhline{2\arrayrulewidth}
Appliance & $\mu$                                       & $\sigma$                                & $\pi$                                    \\ \Xhline{2\arrayrulewidth}
DW        & \begin{tabular}[c]{@{}c@{}}200\\ 400\\ 1000\end{tabular} & \begin{tabular}[c]{@{}c@{}}10\\ 20\\ 50\end{tabular} & \begin{tabular}[c]{@{}c@{}}0.1460\\ 0.0243\\ 0.0609\end{tabular} \\ \hline
RFG       & \begin{tabular}[c]{@{}c@{}}200\\ 400\end{tabular}        & \begin{tabular}[c]{@{}c@{}}20\\ 40\end{tabular}      & \begin{tabular}[c]{@{}c@{}}0.2705\\ 0.0845\end{tabular}     \\ \hline
WD        & 2800                                                        & 37                                                   & 0.0372                                                 \\ \hline
MW        & 1500                                                        & 10                                                    & 0.2272                                                 \\ \hline
KO        & 1070                                                        & 32                                                    & 0.0811                                                \\ \hline
OV        & 4142                                                        & 27                                                    & 0.0426                                                 \\ \hline
BFGI      & 1600                                                        & 12                                                    & 0.0257                                                 \\ \Xhline{2\arrayrulewidth}
\end{tabular}
\end{table}
\subsection{Disaggregation Difficulty Metric}

 The distribution of an event in the aggregated signal given the power distribution of appliances and their participation indices is calculated based on \eqref{p_co} and shown in Fig.~\ref{total_dist}. 
\begin{figure}[t]
    \centering
    \includegraphics[width=.5\linewidth]{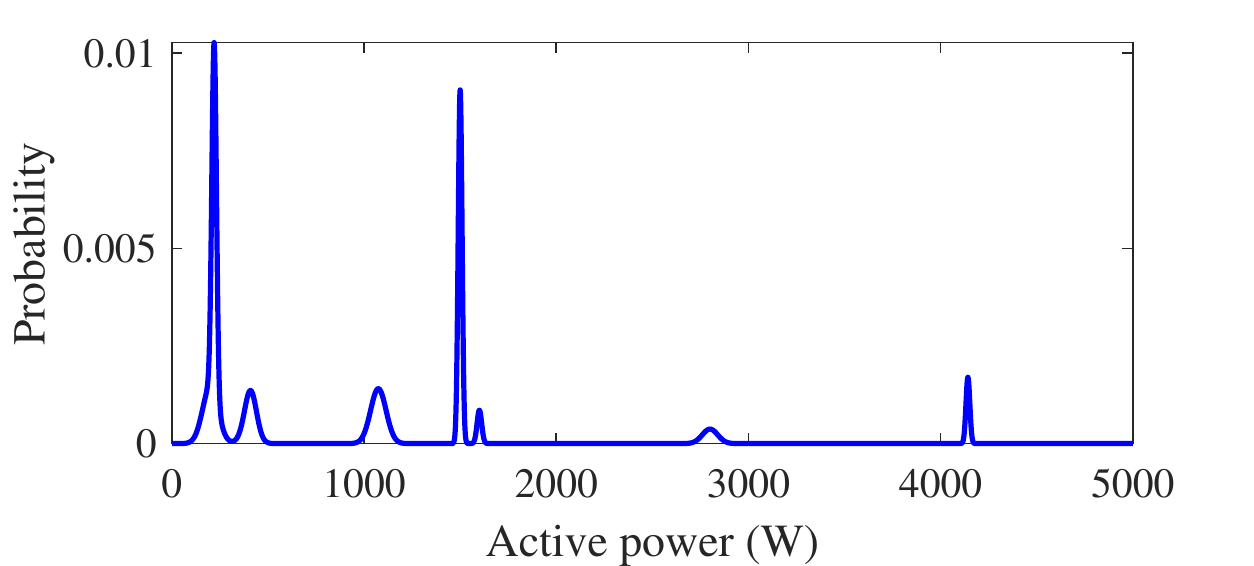}\vspace{.1in}
    \caption{Probability of occurrence of each event given the power distribution of appliances and their participation index}
    \label{total_dist}
\end{figure}

Then, the probability that an event is caused by a specific appliance is obtained from \eqref{belon3} and illustrated in Fig. \ref{belonging}. 
\begin{figure}[t]
    \centering
    \includegraphics[width=.5\linewidth]{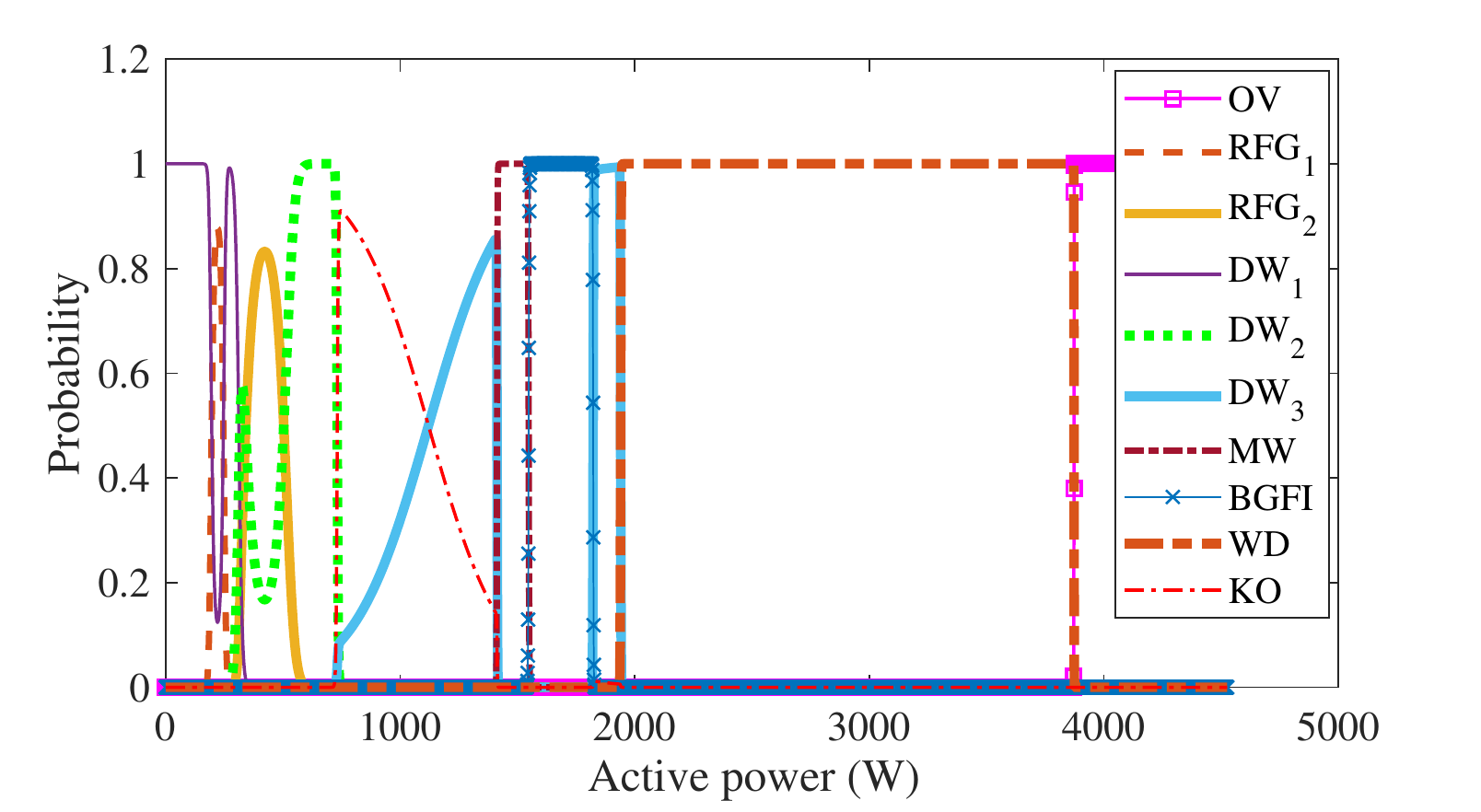}\vspace{.08in}
    \caption{Probability that an event of any power value is caused by each appliance}
    \label{belonging}
\end{figure}
As it is shown, the event with power values which are not in overlapped margins of appliances (such as 4100 W) are more distinguishable and they are assigned to the specific appliance with the probability close to 1. The entropy of each power value in the total signal is computed based on \eqref{DD} and the disaggregation difficulty is computed based on \eqref{8}, which is 0.38 for this set of appliances.
\subsection{Relationship Between the Number of Meters and the DDM}

    In the presence of appliances with close power values in this dataset, event-based NILM classification results in low accuracy. Therefore, with grouping appliances and increasing the number of aggregated signal we can improve the results. 
    The total number of partitions of 7 appliances is 877. Fig.~\ref{DD_sensors} shows the $DDM_{TOT}$ for different number of meters and all 877 types of meter assignment on the appliances. 
\begin{figure}[t]
    \centering
    \includegraphics[width=.5\linewidth]{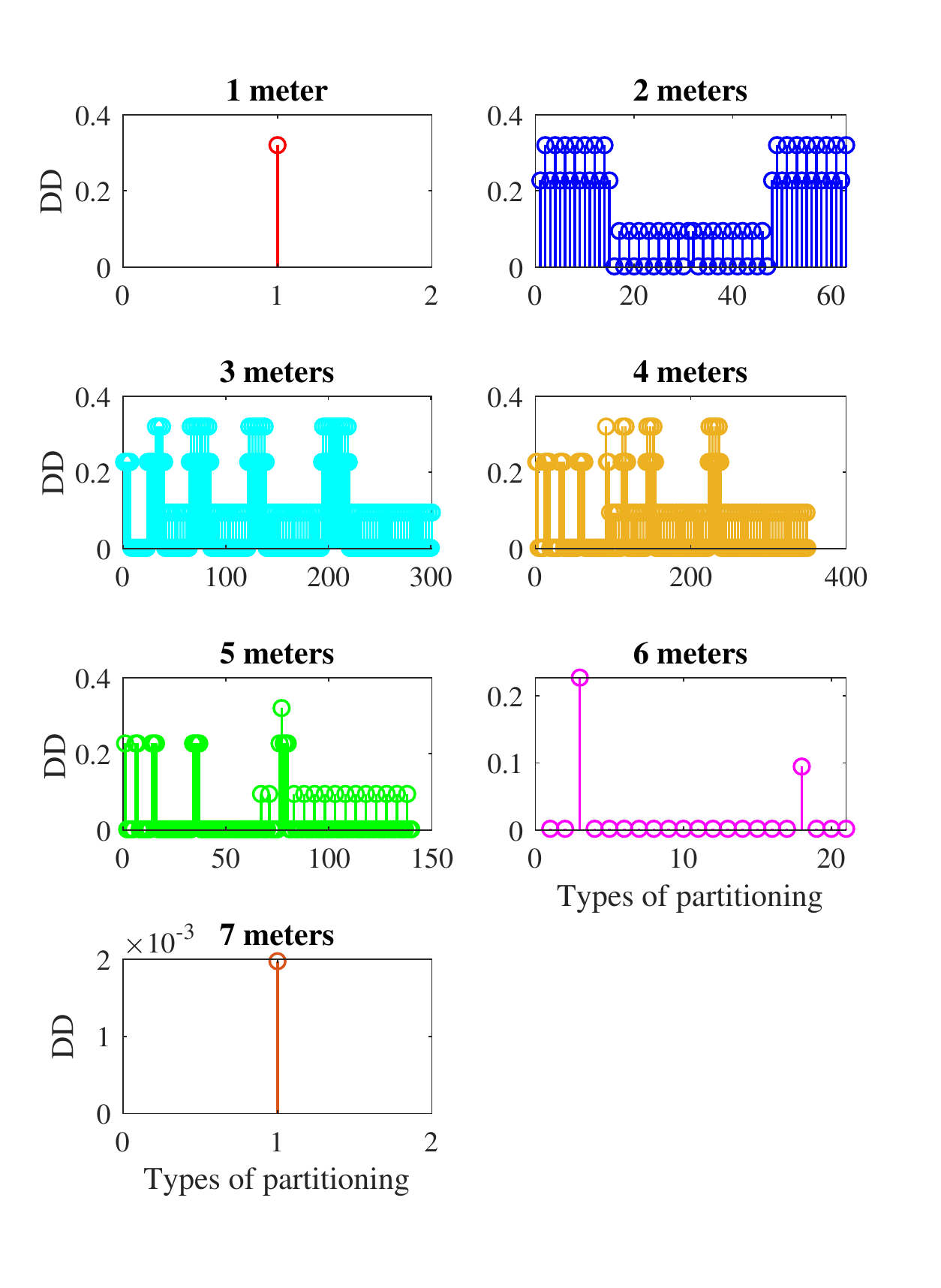}\vspace{-.16in}
    \caption{Disaggregation difficulty vs different types of partitioning}
    \label{DD_sensors}
\end{figure}
\subsection{Discussion}\label{discussion}

    In Subsection~\ref{Feature extraction}, the power distribution of any transition is assumed to be Gaussian for convenience, which may not be a fair assumption in practice. Thus, we now obtain each transition's power distribution via a weighted moving average (WMA) applied on its histogram in the training dataset. The WMA assigns greater weights to the data points, while the weights smoothly fade away as it moves away from the data points. The WMA is calculated by multiplying each observation in the data set by a predefined weighting factor \cite{lucas1990exponentially}. As an example, Fig.~\ref{hist_ma} illustrates the extracted WMA-based power distribution for the microwave. Based on these WMA-based distributions, the DDM is obtained 0.27, which happens to be lower than the DDM given the Gaussian distributions of transitions (DDM=0.32).
\begin{figure}[t]
    \centering
    \includegraphics[width=0.5\linewidth]{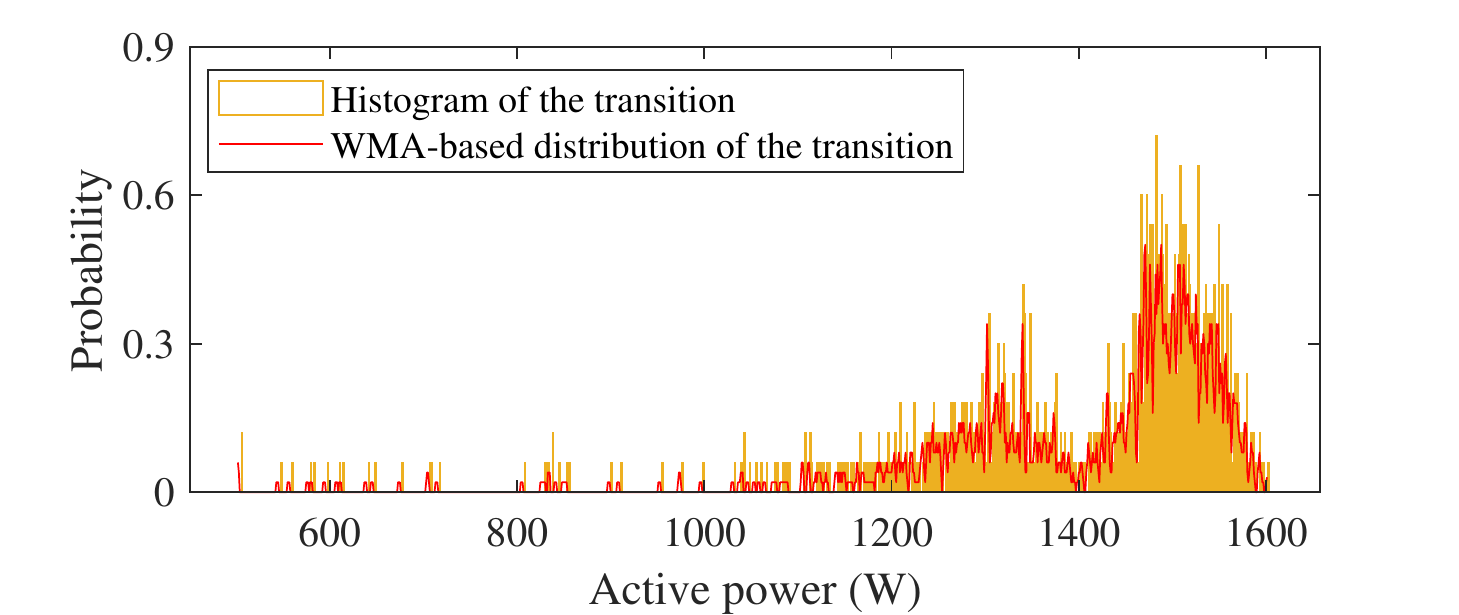}\vspace{.1in}
    \caption{WMA-based distribution of microwave's power values}
    \label{hist_ma}
    \vspace{-.1in}
\end{figure}
\section{Concluding Remarks and Future Work} \label{Conclusion}

    This paper proposed a novel entropy-based concept that quantifies the disaggregation difficulty in event-based classification NILM methods. 
    The proposed method not only takes into account the non-fixed power values of appliances but also adapts to the consumers' usage behavior and the probability of usage of appliances. In the second stage of this paper, the relationship between the number of aggregated meters and the disaggregation difficulty is also quantified. The proper number of meters can be chosen by adapting the existing infrastructure of the building, protecting consumers' comfort and prohibiting extra metering costs. Moreover, based on the proposed method, the effect of deployment of the new appliance which is inherently unavoidable in real-world on disaggregation difficulty can be measured.

Future work will focus on defining disaggregation difficulty based on the total signal without considering the training dataset. Such an approach reduces both instrumentation costs for recording the training dataset and mitigates privacy concerns. Furthermore, we aim to quantify the relation between disaggregation difficulty considering different features of appliances such as reactive power.


\balance
\bibliography{NILM_Second_paper}
\bibliographystyle{IEEEtran}

\end{document}